\documentclass{article}

\usepackage{amsthm}
\newtheorem{theorem}{Theorem}
\newtheorem{lemma}{Lemma}
\newtheorem{remark}{Remark}
\newtheorem{definition}{Definition}

\usepackage[short]{privbib}
\usepackage[english]{babel}
\usepackage{amsmath}
\usepackage{amssymb}
\usepackage{mathrsfs}
\usepackage{xspace}
\usepackage{multirow}
\usepackage{wrapfig}

\newcommand{\xor}{\oplus}
\newcommand{\Xor}{\bigoplus}

\newcommand{\der}{\ensuremath{D}\xspace}

\newcommand{\field}[1]{\mathbb{#1}}
\newcommand{\Z}{\field{Z}}
\newcommand{\R}{\field{R}}
\newcommand{\N}{\field{N}}
\newcommand{\F}{\field{F}}
\newcommand{\Hf}{\ensuremath{\mathcal{H}}\xspace}
\newcommand{\GF}[1]{\ensuremath{\field{GF}\left(#1\right)}}

\newcommand{\indc}[1]{$#1$-wise independence}
\newcommand{\indt}[1]{$#1$-wise independent}

\newcommand{\uindt}[2]{$\left(#1,#2\right)$-wise independent}
\newcommand{\suitable}[1]{$#1$-suitable}
\newcommand{\dermat}[2]{\ensuremath{\mathcal M_{#1}(#2)}\xspace}

\newcommand{\cur}[1]{\ensuremath{\mathscr C_{#1}}}
\newcommand{\arr}[1]{\ensuremath{\mathscr C(#1)}}
\newcommand{\term}[1]{\emph{#1}}

\usepackage{hyperref}

\begin{document}
\title{Independence of Tabulation-Based Hash Classes}

\author{Toryn Qwyllyn Klassen
\thanks{Supported in part by an NSERC USRA (in the summer of 2010) and a PURE award from the University of Calgary (in summer 2011).} 
\and Philipp Woelfel
\thanks{Supported in part by NSERC}
}

\maketitle

\begin{abstract}
A tabulation-based hash function maps a key into $d$ derived characters indexing random values in tables that are then combined with bitwise xor operations to give the hash. Thorup and Zhang \cite{4uni} presented $d$-wise independent tabulation-based hash classes that use linear maps over finite fields to map a key, considered as a vector $(a,b)$, to derived characters. We show that a variant where the derived characters are $a+b\cdot i$ for $i=0,\dots,q-1$ (using integer arithmetic) yielding $(2d-1)$-wise independence. Our analysis is based on an algebraic property that characterizes $k$-wise independence of tabulation-based hashing schemes, and combines this characterization with a geometric argument. We also prove a non-trivial lower bound on the number of derived characters necessary for $k$-wise independence with our and related hash classes.
\end{abstract}

\section{Introduction}
A family (multiset) $\Hf$ of hash functions $h:U\to R$ is called \indt{k} if a hash function $h\in\Hf$ selected uniformly at random maps any $k$ distinct keys from $U$ uniformly and independently to $R$.
Such classes of functions have found wide application in the literature.
For example, \indt{4} hash functions can be used for estimating the second moment of a data stream (see Thorup and Zhang \cite{4uni}).
Pagh, Pagh, and Ru\v{z}i\'{c} \cite{paghlinear} proved that insertions, queries, and deletions using hashing with linear probing will all run in expected constant time if the hash function used is \indt{5} (assuming the number $n$ of keys hashed is at most a constant fraction of the table size). 
Seidel and Aragon \cite{rst96} introduced the treap data structure and showed that if the priorities of a set of keys are 8-wise independent random variables, then various performance guarantees hold \cite[Theorem 3.3]{rst96}.

For most applications, \emph{approximate} \indc{k}, where the probability that $k$ keys map to $k$ given values may deviate by a small amount from the true random case, suffices.
The canonical approximate \indt{k} hash family (where the range is $[m]:=\{0,1,\dots,m-1\}$), is constructed by choosing a prime $p\geq m$ and taking as the family the set of all mappings $x\mapsto P(x)\bmod m$, where $P$ is a polynomial of degree $k-1$ over $\Z_p$.

Evaluating the polynomial is usually inefficient, especially for large values of $k$. 
Therefore, there has been interest in faster methods trading off arithmetic operations for a few lookups in tables filled with random values.

\subsection{Background on Tabulation-Based Hashing}
The idea behind tabulation-based hashing is splitting a key $x$ into characters (substrings) $x_0,\dots,x_{q-1}$ which are then hashed by $q$ hash functions into a common range.
The resulting values are then combined with exclusive or operations to yield the final hash value, $h(x)$.
(For simplicity we assume for now that the range of $h$ is the set of all bit strings of some fixed length.)
Such a scheme was proposed by Carter and Wegman \cite{Carter1979}.

\sloppypar
Since characters are shorter than the key, it can become feasible (and desirable) to tabulate the hash functions on the characters; that is, getting the hashed value of a character is just a lookup into a table position indexed by that character.
A function from such a hash class is therefore given by $h(x_0\dots x_{q-1})=\Xor_{i\in[q]} T_i(x_i)$, where $\xor$ denotes the bitwise exclusive or operation, $[q]:=\{0,\dots,{q-1}\}$, and $T_i$, for each $i\in [q]$, is a table filled with random values.
Selecting $h$ uniformly at random selects the tables randomly; it is known that this scheme is \indt{3} if the tables are filled with \indt{3} random values, but irrespective of the tables' contents, $4$-wise independence cannot be achieved.
(Obviously, for any four 2-character keys $(a,c),(a,d),(b,c),(b,d)$, the hash value of one key is uniquely determined by that of the three other keys.)

A modification to the scheme that can achieve $k$-wise independence for $k\geq 4$ is to \emph{derive} a small number of characters $D_0(x),D_1(x),\dots,D_{d-1}(x)$ from a key $x$, and use these derived characters in place of the components of the key for table lookups.
That is, a hash function would be given by $h(x)=\Xor_{i\in[d]} T_i(D_i(x))$, where $T_0,\dots,T_{d-1}$ are tables filled with random values.
Dietzfelbinger and Woelfel \cite{randomgraphs} suggested to use random hash functions $D_i$ (from a $c$-universal hash family) in order to derive characters.

Thorup and Zhang \cite{4uni} proposed an efficient deterministic way of computing derived characters:
Consider an arbitrary finite field $\F$.
Let $\vec x=(x_0,\dots,x_{q-1})$ be a vector over $\F$, and $G$ a $q\times d$ matrix (where $d>q$) such that every $q\times q$ submatrix has full rank over $\F$.
Then $\vec z = \vec x G$ forms a vector $\bigl(D_0(\vec x),\dots,D_{d-1}(\vec x)\bigr)$ of $d$ derived characters.
Thorup and Zhang proved that such a class is $k$-wise independent if $d=(k-1)(q-1)+1$.
In fact, it suffices that each of the tables $T_0,\dots,T_{k-1}$ is filled independently with $k$-wise independent values.
If the input characters are $c$-bit strings, then one can choose $\F=\GF{2^c}$ to obtain $c$-bit strings as derived characters.
Thorup and Zhang also gave a similar scheme for \indt{4} hash classes (which they later \cite{5uni} proved was also \indt{5}) that required only $2q-1$ derived characters for $q$ input characters (as a tradeoff, some of the derived characters were slightly larger than the input characters). 
A version of the \indt{5} scheme for $q=2$ is $h(ab)=T_0(a)\xor T_1(b)\xor T_2(a+b)$, where `+' is either regular integer addition or addition modulo a suitable prime.

\subsection{Contributions}
The efficiency of tabulation based hashing is heavily influenced by the number of table lookups, and thus by the number of derived characters.
Therefore, we study how many derived characters are needed to obtain $k$-wise independent tabulation based hash functions.

We suggest a variant of Thorup and Zhang's tabulation based hash families, called $(q,d)$-curve hash families.
For $q=2$, these hash families can achieve $(2d-1)$-wise independence using only $d$ derived characters.
Thus, for $k\geq 6$, only about half as many table lookups in random tables, albeit slightly larger ones, are needed as in Thorup and Zhang's construction in order to achieve guaranteed $k$-wise independence.
Another advantage of our construction is that these hash functions can be computed with simple integer arithmetic (the \mbox{$i$-th} derived character for a key $x=(a,b)$ is simply $a+i\cdot b$), as opposed to finite field arithmetic.
In fact, the only way Thorup and Zhang's scheme can achieve practical performance is by using multiplication tables for the finite field multiplications.
In our scheme, multiplication tables are not necessary.

Generally, a function from our $(q,d)$-curve hash family maps a key $a_0{\dots}a_{q-1}$ to the $y$-values attained by the polynomial curve $y=\sum_{i\in[q]} a_i z^i$ at ${z=0,\dots,d-1}$.
These $y$-values form the $d$ derived characters that can be used for lookups in the random tables $T_0,\dots,T_{d-1}$.
(As in other tabulation based hash classes, we only require that the tables $T_0,\dots,T_{d-1}$ are filled independently with $k$-wise independent random values.)
Using integer arithmetic allows for the question of whether a $(q,d)$-curve hash family is $k$-wise independent to be easily interpreted as a geometric problem regarding what intersections occur in an arrangement of curves.

When $q=2$, the polynomials are linear, and the problem is especially simple. We have the following result:
\begin{theorem}\label{upper}
A $(2,d)$-curve hash family is \indt{(2d-1)}.
\end{theorem}

We also establish a lower bound on the number of derived characters that are needed for $k$-wise independence with a $(2,d)$-curve hash family.

\begin{theorem}\label{lower}
No $(2,d)$-curve hash family on $U=[n]^2$ is \indt{\left(2^d\right)}, provided that $n\ge\max\{2^{d-1}(d-1)+2,3\}$.
\end{theorem}

We prove these theorems through geometric arguments.
At the heart of our analysis is an algebraic characterization of the tabulation-based hash classes which are \indt{k}.
This characterization is of independent interest.
For example, it can be used to establish that every \indt{2k} tabulation-based hash class (based on fully random tables) is also \indt{(2k+1)}.
This has the immediate consequence, that for odd $k>1$, in Thorup and Zhang's construction $(k-2)(q-1)+1$ derived characters are sufficient for $k$-wise independence, which is $q-1$ fewer characters than what their proof guarantees.
In addition, the algebraic characterization simplifies the analysis of tabulation-based hash classes.

For $q>2$, we have been able to achieve only a small reduction in the number of required derived characters compared to Thorup and Zhang's hash families: We show that in general, $(q,d)$-curve hash functions are $k$-wise independent if
$d\ge\left\lceil2\frac{q-1}{2q-1}(k-1)\right\rceil(q-1)+1$, whereas Thorup and Zhang's method requires $d\geq (k-1)(q-1)+1$.
For values of $k$ where we don't have to round up, the reduction in lookups needed is only by a small constant factor that decreases with $q$ (e.g., $4/5$ for $q=3$ or $6/7$ for $q=4$).
Although we do not expect this improvement to be relevant in practice (especially since using integer arithmetic means that larger lookup tables are needed), we believe that the theoretical result indicates that perhaps improved proof techniques can lead to further improvements in the future.

We have executed some initial experiments (see the appendix).
These show that $(2,d)$-curve hash families outperform (often by more than a factor of 2) the functions from Thorup and Zhang's $q=2$ class that are known to give at least the same degree of independence.
However, for $q=4$, Torup and Zhang's functions are still more efficient than our $(2,d)$-curve functions, presumably due to the fact that the random tables $T_i$ are smaller and exhibit a more cache friendly behaviour.
Nevertheless, our experiments indicate that reducing the number of table lookups can significantly increase the efficiency of tabulation based hashing, and that it is worthwhile trying to determine the exact independence of such hash classes.

\section{The Independence of Tabulation-Based Hash Classes}
The following definition of $k$-wise independent hash classes is standard:
\begin{definition}
A class $\Hf$ of hash functions of the form $h:U\to[m]$ is called \term{\indt{k}} (where $k\in \N$), if
$\Pr(\forall i \in [k]:h(x_i)=y_i)=\frac{1}{m^k}$
 for all distinct $x_0,\dots,x_{k-1}\in U$
 and all $y_0,\dots,y_{k-1}\in [m]$,
and $h$ selected uniformly at random from \Hf.
\end{definition}
Throughout this paper we will assume that $m=2^\ell$ for some $\ell\in\N$ (note that we are denoting the positive integers by $\N$; the non-negative integers will be called $\N_0$).
It is known and easy to see that every \indt{k} class is \indt{k'} if $1\le k'\le k$. A hash function is called \indt{k} if it is selected uniformly at random from a \indt{k} class.

Now we will detail tabulation-based hashing and the terminology involved. This description is based on the ``general framework'' described by Thorup and Zhang \cite{4uni,5uni} and is also influenced by some notation used by P{\v a}tra{\c s}cu and Thorup \cite{Patrascu}.

A \term{derivation function} $D$ maps a key $x$ to a sequence $\der(x):=(\der_i(x):i\in[d])=\{(i,\der_i(x)):i\in[d]\},$ where each $\der_i$ is some function. This sequence is called a \term{derived key}; the element $\der_i(x)$ is called the $i$-th \term{derived character}. A \term{tabulation-based hash function} $h:U\to [m]$, using derivation function \der, is given by
\begin{align*}
h(x)
=\Xor_{i\in [d]} T_i(\der_i(x))
\end{align*}
where each $T_i$, $i\in [d]$, is some hash function into $[m]$.
A tabulation-based hash class $\Hf$ is a multiset of tabulation-based hash functions where each member has (possibly) different functions $T_i$, $i\in[d]$. We will use the notation $\Hf_\der$ to denote the tabulation-based hash class whose members use derivation function $\der$.

The intention of a tabulation-based hash function is that the computation of each $T_i$ is just one table lookup.
The idea is that for each $i\in[d]$, $D_i(U)=\{D_i(x):x\in U\}$ is a subset of $[n_i]$ for some small $n_i$, so that those tables are small enough to fit in fast memory.

\begin{definition}
A tabulation-based hash class $\Hf$ is called \term{\suitable{k}} if each table $T_i$, $0\leq i<d$, is filled with $k$-wise independent random values, and the choices for $T_0,\dots,T_{d-1}$ are independent.
\end{definition}

Some previous considerations of what degree of independence are achieved by tabulation-based hash functions have made use of the following result:
\begin{lemma}[{\cite[Lemma 2.1]{4uni}, see also \cite[Lemma 2.6]{Siegel04}}]\label{peeling}
A \suitable{k} tabulation-based hash function is \indt{k} if for each set of keys $S'$ of size $k' \le k$, there exists an $i \in [d]$ and $x \in S'$ such that $D_i(x) \not= D_i(y)$ for any $y \in S'\setminus\{x\}$.
\end{lemma}

We will generalize this result to get a characterization of which \suitable{k} tabulation-based hash functions are \indt{k}. First, we make one more definition.

\begin{definition}For any derivation function $\der$ and set $S=\{x_0,\dots,x_{k-1}\}$ of $k$ keys, the \term{derivation incidence matrix} \dermat{\der}{S}
is a $(0,1)$-matrix having a column corresponding to each element in $\bigcup_{j\in[k]}\der(x_j)$, and having 1 in row $j$ and column $(i,a)$  if and only if $\der_i(x_j)=a$.
\end{definition}
To illustrate, if $S=\{x,y,z\}$ and $x$ has derived key $\{(0,4),(1,5),(2,6)\}=(4,5,6)$ and similarly $\der(y)= (4,7,8),\ \der(z) = (5,7,9)$, then $\dermat{\der}{\{x,y,z\}}$ is
{\small
\[\bordermatrix{
    & (0,4) & (0,5) & (1,5) & (1,7) & (2,6) & (2,8) & (2,9) \cr
    x & 1 & 0 & 1 & 0 & 1 & 0 & 0 \cr
    y & 1 & 0 & 0 & 1 & 0 & 1 & 0 \cr
    z & 0 & 1 & 0 & 1 & 0 & 0 & 1 \cr
}\]}
The derivation incidence matrix is unique up to reordering the rows and columns.

We show that $k$ keys from a set $S$ are mapped uniformly and independently by a random hash function from $\Hf_\der$, if and only if the rows in \dermat{\der}{S} corresponding to those $k$ keys are linearly independent.
The idea of the following theorem is very similar to a proof by Dietzfelbinger and Rink \cite[Proposition 1]{splittingtrick} where a hash function was shown to be fully random.
\begin{theorem}\label{rankTok}
  Let $\Hf_\der$ be a \suitable{k} tabulation-based hash class.
  Then $\Hf_\der$ is \indt{k} if and only if for every set $S$ of $k$ keys, \dermat{\der}{S} has full row rank over \GF{2}.
\end{theorem}
\begin{proof}
Let $S=\{x_0,\dots,x_{k-1}\}$ be an arbitrary set of keys, and let $w$ be the number of columns in $M=\dermat{\der}{S}$.

First, suppose that $M$ has full row rank over \GF{2} (and so over \GF{2^\ell}=\GF{m}).
Pick $h$ from a \suitable{k} tabulation-based hash class $\Hf_\der$ uniformly at random, and define a vector $V=(v_0,\dots,v_{w-1})^T\in[m]^w$ such that if the $i$-th column of $M$ is labelled $(j,a)$, then $v_i=T_j(a)$. Then the matrix multiplication $M\cdot V$ results in a $k\times 1$ column vector whose $i$-th entry is $h(x_i)$.

Note that since there are $k$ keys, there are at most $k$ entries in $V$ that are based on the values in the table $T_j$ for each $j\in[d]$. Since the tables represent independently chosen \indt{k} functions, this means that $V$ is distributed uniformly in $[m]^w$.

We consider the probability that $M\cdot V=Y$, for a fixed but arbitrary $Y\in[m]^k$. Since $M$ has full row rank there is at least one solution $V_0\in[m]^w$ to that equation, and the set of all solutions is $\{V_0+Z:Z\in \operatorname{Ker}\left( M\right)\}$ and so has dimension $w-k$ by the rank-nullity theorem.

Therefore, since $V$ is selected uniformly at random from $[m]^w$, $$\Pr(M\cdot V=Y)=\frac{m^{w-k}}{m^w}=\frac{1}{m^k}.$$
So $\Hf_\der$ is \indt{k}.

Now suppose that \dermat{\der}{S} does not have full row rank.
Then there is a row, w.l.o.g. row $k-1$, which is a linear combination of the others. Let $\Psi$ be the set of the indices of the rows whose sum is equal to the last row.
It follows that $\Pr(h(x_i)=y_i, \forall i\in [k])=0$
for any $h\in\Hf_\der$, $h(x_{k-1})=\Xor_{i\in\Psi} h(x_i)$. So for any $y_0,\dots,y_{k-2}$ and $y_{k-1}\ne \Xor_{i\in\Psi}y_i$,
Therefore, $\Hf_\der$ is not \indt{k}.
\end{proof}

We can see that \autoref{rankTok} is a generalization of \autoref{peeling}. Suppose the condition for \autoref{peeling}, that for each set $S'$ of $k'$ keys there are $i\in [d], x\in S'$ such that $D_i(x)\not= D_i(y)$ for $y\in S'\setminus\{x\}$. Hence for each such $S'$ with the corresponding appropriate $i$ and $x$, in \dermat{\der}{S'} the column labelled $(i,D_i(x))$ contains exactly one 1 (and so all rows do not sum to zero). Therefore, there is no subset of the rows of \dermat{\der}{S} that can sum to zero, i.e. \dermat{\der}{S} has full row rank.

\subsection{Weaker independence}
Various properties of hash families weaker than \indc{k} have been considered, such as the requirement that for all distinct $x_0,\dots,x_{k-1}\in U$ and all $y_0,\dots,y_{k-1}\in [m]$,
and $h$ selected uniformly at random from \Hf,
\begin{align*}
\Pr(h(x_i)=y_i,\ \forall i\in[k]) \le \frac{c}{m^k}.
\end{align*}
We will call such a class \uindt{k}{c} (after \cite{Siegel04}). It may be worth noting that some of the lemmas we have just given can be adapted to describe \uindt{k}{c} families.

If we modify the conditions of \autoref{rankTok} to only require that the tables of a hash function $h\in\Hf_\der$ selected uniformly at random are independent \uindt{k}{c} functions, then the probability of the vector $V$ defined in the proof of \autoref{rankTok} assuming a particular value is bounded above by $\frac{c^d}{m^k}$. Therefore $\Hf_\der$ is \uindt{k}{c^d}.

If \dermat{\der}{S} does not have full row rank, then as in the latter part of the proof for \autoref{rankTok} let us suppose that the row corresponding to $x_{k-1}$ is equal to the sum of the rows corresponding to the keys indexed by $\Psi$. Then for any $h\in\Hf_\der$ at least one of the $m^{k-1}$ possible $(k-1)$-tuples $y_0,\dots,y_{k-2}$ must be equal to $(h(x_0),\dots,h(x_{k-2}))$ with probability at least $\frac{1}{m^{k-1}}$, so fixing such a tuple we get
\begin{align*}
\Pr\left(h(x_i)=y_i,\ \forall i\in[k-1] \wedge h(x_{k-1})=\Xor_{i\in\Psi} h(x_i)\right) \ge \frac{1}{m^{k-1}}=\frac{m}{m^k}
\end{align*}and hence $\Hf_\der$ can not be \uindt{k}{c} for $c<m$.

\subsection{Even Degrees of Independence Imply Odd Degrees}
There is an interesting result that follows from the characterization of \indt{k} hash classes showing that achieving \indc{k} for odd $k$ does not require more derived characters than achieving $(k-1)$-wise independence.

\begin{lemma}\label{oddeven}
If $\Hf_\der$ is a \indt{2k} tabulation-based hash class and is \suitable{(2k+1)}, then $\Hf_\der$ is \indt{(2k+1)}.
\end{lemma}
\begin{proof} 
Suppose for contradiction that $\Hf_\der$ is not \indt{(2k+1)} even though it is \suitable{(2k+1)}. 
Then by Theorem~\ref{rankTok} there exists a set of $2k+1$ keys $S$ such that \dermat{\der}{S} does not have full row rank over \GF{2}; that is, some linear combination of the rows sum to zero. This is equivalent to saying that there is a subset $\Psi$ of the rows that sum to zero.

Since $\Hf_\der$ is \indt{2k}, $|\Psi| > 2k$. Therefore, for $\Hf_\der$ not to be \indt{(2k+1)}, the sum of all rows of \dermat{\der}{S} must be zero.

Let $M$ be the submatrix of $\dermat{\der}{S}$ containing only those columns corresponding to the first derived characters (i.e., those columns labeled with $(0,a)$ for some $a$). The entries in each column of $M$ must sum to zero (modulo 2), so the number of ones in each of these columns, and therefore in $M$, must be even. However, there are exactly $2k+1$ ones (a single one for each key) distributed among these columns, so this is impossible.
\end{proof}

This generalizes and simplifies the proof of the result by P{\v a}tra{\c s}cu and Thorup \cite{Patrascu} that any \indt{4} tabulation-based hash class in which all input characters are used as derived characters is \indt{5}. 
Thorup and Zhang's tabulation-based construction from \cite{4uni} achieves $k$-wise independence for $d=(k-1)(q-1)+1$.
Lemma~\ref{oddeven} implies that their construction is \indt{k} for $d=(k-2)(q-1)+1$ if $k>1$ is odd.

\section{$(q,d)$-curve Hash Families}

We will consider a set of hash classes that are variants of the \indt{k} scheme of Thorup and Zhang \cite{4uni}.
\begin{definition}For each $q\in\{2,3,4,\dots\},d\in \N$, a $(q,d)$-curve family of hash functions is a tabulation-based hash family with derivation function $\der$ given by
\begin{align*}
\der_i(a_0a_1\dots a_{q-1})=\sum_{r\in[q]}a_r\cdot i^r
\end{align*}
for $i\in [d]$.
\end{definition}

According to this definition, each key $a=a_0,\dots,a_{q-1}$ determines a polynomial curve in the plane.
The $j$-the derived character of $a$, $D_j(a)$, is then the $y$-value of the curve at $x$-coordinate $j$.
This motivates the following definitions, which are intended to aid geometric reasoning.
\begin{definition}
For any key $a=a_0\dots a_{q-1}$, the corresponding \term{key curve} $\cur{a}:\R\to\R$ is defined by $\cur{a}(z)=\sum_{i\in [q]} a_i\cdot z^i$.
Given a set $S$ of keys, we will use the notation \arr{S} to denote the set $\{\cur{a}:a\in S\}$.
\end{definition}

\begin{definition}
A \term{column} is a set $\{(c,y):y\in \Z\}=\{c\}\times\Z$ for some $c\in\Z$.
\end{definition}
\begin{definition}
A \term{bad column} relative to a set of keys $S$ is a column $\Psi$ such that each point in $\Psi$ is intersected by an even (possibly zero) number of elements of $\arr{S}$. That is, for all $(c,y)\in\Psi$ the cardinality of $\{a\in S: \cur{a}(c)=y\}$ is even.
\end{definition}
\begin{definition}
A \term{bad $(q,d,k)$-arrangement} over a set $U$ is a set of $k$ key curves (corresponding to keys in $S\subseteq U$) derived using a $(q,d)$-curve hash family and having $d$ consecutive bad columns $\{0\}\times\Z$, $\{1\}\times\Z$,\dots,$\{d-1\}\times\Z$.
\end{definition}

For an arbitrary set $S$ of $k'$ keys, $\arr{S}$ is a bad $(q,d,k')$-arrangement if and only if there is an even number of ones in each column of \dermat{\der}{S} (i.e., \dermat{\der}{S} does not have full row rank).
Hence, by Theorem~\ref{rankTok}, a \suitable{k} $(q,d)$-curve family of functions mapping from a universe $U$ is \indt{k} if and only if for every $k'\in\{1,\dots, k\}$, there is no bad $(q,d,k')$-arrangement over $U$. 

Let $k_{\max}(q,d)$ denote the largest integer $k$ such that for any $k'\leq k$ there is no bad $(q,d,k')$-arrangement over $(\N_0)^q$.
By the discussion above, for $k=k_{\max}(q,d)$, a $(q,d)$-curve family is \indt{k} if it is \suitable{k}, but cannot be \indt{(k+1)} if the universe it acts on is large enough to include the set of keys corresponding to a bad $(q,d,k+1)$-arrangement.
In the following we determine upper and lower bounds on $k_{\max}(2,d)$.

\subsection{A Lower Bound for $k_{\max}(2,d)$}
Based on the geometric observations we have just made, to prove \autoref{upper} it suffices to prove the following lemma.

\begin{lemma}
  For any $d\in\N$, and $k\in\{1,\dots,2d-1\}$, there does not exist a bad $(2,d,k)$-arrangement over $(\N_0)^2$, i.e., $k_{\max}(q,d)\geq 2d-1$.
\end{lemma}
\begin{proof}
Let $S=\{a_0b_0,\dots, a_{k-1}b_{k-1}\}$ be an arbitrary set of $k\in\{1,\dots,2d-1\}$ keys in $(\N_0)^2$.
If $k$ is odd then the statement is trivial, since a bad column is one that is intersected at every point by an even number of curves, which is not possible when the total number of curves intersecting the column is odd.

Now let $k$ be even and w.l.o.g. say that $b_0 \ge b_i$ for all $i \in[k]$.
Suppose for contradiction that $\arr{S}$ is a bad $(2,d,k)$-arrangement.

For each $c\in[d]$, $\arr{S}$ can be partitioned into $\{A_c,E_c,B_c\}$, the sets of curves that are respectively, above, equal to, and below $\cur{a_0b_0}$ in column $c$.
More precisely,
$
A_c:=\{f\in\arr{S}: f(c)>a_0+b_0c\}$,
$E_c:=\{f\in\arr{S}: f(c)=a_0+b_0c\}$, and
$B_c:=\{f\in\arr{S}: f(c)<a_0+b_0c\}$.

Since $\arr{S}$ is a bad arrangement, 
each of these subsets must have an even cardinality. Note that for each $c\in[d-1]$, if $f\in B_c$, then $f\in B_{c+1}$, for since $\cur{a_0b_0}$ is a line with a slope ($b_0$) at least as great as any other, no line that is below $\cur{a_0b_0}$ at some $c$ can ever rise above later. So $|B_{c+1}|\ge |B_c|$ for all $c\in[d-1]$.

For each $c\in [d-1]$, $|E_c|$ is even so there exists $f\in E_c\setminus \{\cur{a_0b_0}\}$. Since lines can only intersect once and $\cur{a_0b_0}$ has slope at least as great as any other line, $f\in B_{c+1}$. So $|B_{c+1}|> |B_c|$. Since the cardinality of $B_{c+1}$ is also even, that means that $|B_{c+1}|\ge |B_c|+2$, and since $|B_0|\ge 0$, $|B_{d-1}|\ge 2(d-1)=2d-2$. For this to be possible, we must require $k>2d-2$. But then $k$ can only be $2d-1$ which is odd.
\end{proof}

\subsection{An Upper Bound for $k_{\max}(2,d)$}

We now consider the problem of determining how many derived characters are needed for \indc{k} with a $(q,d)$-curve hash family. The following result shows that, under certain conditions, if we want to double the degree of independence we get from such a class, then the required value of $d$ must increase by at least $q-1$.

\begin{lemma}\label{curvelower}
If there is a bad $(q,d,k)$-arrangement over $\Z$, then there is a bad $(q,d+q-1,2k)$-arrangement over $\Z$.
\end{lemma}

\begin{proof}
Suppose that we have a bad $(q,d,k)$-arrangement $C=\arr{S}$. We will use this in constructing a bad $(q,d+q-1,2k)$-arrangement.

Let $Q(z)$ be a polynomial of degree $q-1$ having zeros at each $z\in\{d+i:i\in[q-1]\}$ and such that
$
C':=\{P(z)+Q(z):P(z)\in C\}
$ 
is disjoint from $C$. Let us say that $C'=\arr{S'}$ for some set of keys $S'$.

Note that $C'$ is, like $C$, bad on columns $0,\dots, d-1$ since for any $P_1,P_2\in C$, 
$P_1(z)=P_2(z)$ if and only if $P_1(z)+Q(z)=P_2(z)+Q(z)$.
Then let $C''=C\cup C'=\arr{S\cup S'}$. Thus $C''$ has $2k$ curves, and as we will show, is a bad $(q,d+q-1,2k)$-arrangement.

Consider column $z$, for arbitrary $z\in[d]$. Since both $C$ and $C'$ are bad on column $z$, any point $(z,y)$ is on an even number $c$ of curves from $C$ and an even number $c'$ of curves from $C'$. 
Therefore the number of curves from $C''$ that pass through $(z,y)$ is $c+c'$ which is also even.

Columns $d$ through $d+q-2$ are also bad, since for any $z\in \{d,\dots, d+q-2\}$, $P(z)=P(z)+Q(z)$ for all $P\in C$. Every point in columns $d$ through $d+q-2$ is on an equal number of curves from $C$ and $C'$.
\end{proof}

To apply \autoref{curvelower} to describe the behavior of a hash family we need to know that the bad $(q,d+q-1,2k)$-arrangement constructed in the proof actually corresponds to keys in the universe of the hash functions in that class. The following result, in showing sufficient conditions on the universe size in the case of $q=2$, implies \autoref{lower}.

\begin{lemma}\label{lem:bad-arrange}
For each $d\in\N$ there is a set $S$ of $k=2^d$ keys in $U=[n]\times[n]$ such that $\arr{S}$ is a bad $(2,d,2^d)$-arrangement, provided that $n\ge \max\left\{2^{d-1}(d-2)+2,\ 3\right\}$.
In particular, $k_{\max}(2,d)<2^d$.
\end{lemma}

\begin{proof}\def\x{z}
We will construct, for each $d\ge 3$, a bad $(2,d,2^d)$-arrangement where each input character is an element of $\left[2^{d-1}(d-2)+2\right]$ (for the special cases of $d \in\{1,2\}$ we will have arrangements where we require the input characters to be in $[3]$). From this it follows that $(k=2^d)$-independence cannot be achieved with $d$ derived characters, assuming that the universe is large enough to have characters of the appropriate size.

The proof is by induction on $d$. The base cases are $d=1,2,3$.
The arrangements $\{0+0{\x}, 0+1{\x}\}$ and $\{0+1{\x},0+2{\x},1+0{\x}, 1+1{\x}\}$ of $2^1$ and $2^2$ lines are bad on $d=1$ and $d=2$ columns respectively, and use input characters from $[3]$ only.
For $d=3$, the set of $2^3=8$ key curves $\{0+3{\x},0+4{\x},1+2{\x},1+3{\x},4+1{\x},4+2{\x},5+0{\x},5+1{\x}\}$ is bad in three columns and so is a bad $(2,3,2^3)$-arrangement.
Note that all input characters are in $[2^{d-1}(d-2)+2]=[2^2+2]=[6]$.

Suppose that there is a bad $(2,d,2^d)$-arrangement $C$, for some $d\ge 3$, derived from keys where all the input characters are in $[2^{d-1}(d-2)+2]$. We now apply the construction in \autoref{curvelower} with $Q({\x})=2^{d-1}(d-{\x})$ to get a bad $(2,d+1,2^{d+1})$-arrangement where the input characters are in $[2^d(d-1)+2]$.

First, define $C':=\{(a+b{\x})+2^{d-1}(d-{\x}):(a+b{\x})\in C\}$.
We want $C\cap C'=\emptyset$: for each $a+b{\x}\in C,\ a\le 2^{d-1}(d-2)+1$ whereas for each $a'+b'{\x}\in C',\  a'\ge 2^{d-1}d>2^{d-1}(d-2)+1$.

Therefore, $C''=C\cup C'$ is a set of $2^{d+1}$ key curves, derived from keys where the first input characters are in $[(2^{d-1}(d-2)+2)+2^{d-1}d]=[2^{d}(d-1)+2]$ and the second input characters are in $\{-2^{d-1},\dots,2^{d-1}(d-2)+1\}$.By increasing the value of all second input characters by $2^{d-1}$ we ensure that they are all in $[2^{d-1}(d-1)+2]\subset[2^{d}(d-1)+2]$. This clearly does not change which columns are bad.
\end{proof}

\begin{remark}If we consider a hash class $\Hf$ which is a variant of $(q,d)$-curve hash family in that the derived characters are reduced modulo $r$ for any natural number $r$ (not necessarily a prime), then any set $S$ of keys for which $\arr{S}$ would be a bad $(q,d,k)$-arrangement will not be hashed with \indc{k} by $h$ selected uniformly at random from $\Hf$.

Therefore, any lower bound on the number of derived characters needed for \indc{k} with a $(q,d)$-curve hash family would also apply to $\Hf$.
\end{remark}

\subsection{A Lower Bound On $k_{\max}(q,d)$ for $q>2$}
We now consider how many derived characters are sufficient for \indc{k} for $(q,d)$-curve hash families where $q>2$.

\begin{lemma}If $C$ is a bad $(q,q,k)$-arrangement, there must be at least $k/2$ intersections between two curves on each column, and further to these there must be at least $\lceil k/4 \rceil$ additional intersections that are each either between columns or else on a column other than the last one.
\end{lemma}
\begin{proof}\def\x{z}
Let $C$ be a bad $(q,q,k)$-arrangement, and choose a partition of $C$ into pairs $(f_i,g_i),i\in[k/2 ]$ such that $f_i(0)=g_i(0)$ (it is possible to choose such intersecting pairs since the arrangement is bad). Furthermore, since every column is bad, we can choose $k/2$ intersections for each column such that each curve is in exactly one of those intersections.

We will associate with each pair $(f_i,g_i)$, $i\in[k/2 ]$, an additional intersection (involving one or both of its members)
that occurs before the last column. So we will have $k/2$ such associations, and therefore at least $\lceil(k/2)/2\rceil=\lceil k/4 \rceil$ additional intersections (since an intersection of two curves may involve curves from at most two associations).

Fix an arbitrary pair $(f_i,g_i)$.
For each ${\x}\in[d]$, let $A_{\x},E_{\x},$ and $B_{\x}$ be the subsets of curves that are respectively greater than, equal to, and less than $f_i$ in column ${\x}$.
Since the arrangement is bad, each of these sets has even cardinality for each ${\x}$ (and since $f_i\in E_{\x}$ for each ${\x}$, $|E_{\x}|\ge 2$). If $|E_{\x}| >2$ for some ${\x}\in[q-1]$, then there is some $\alpha \in E_{\x}\setminus\{f_i\}$ that we can pick to associate the intersection of $\alpha$ and $f_i$ at ${\x}$ with the pair $(f_i,g_i)$.

Otherwise, $|E_{\x}|=2$ for all ${\x}\in[q-1]$. By construction, $E_0=\{f_i,g_i\}$. Since $f_i,g_i$ are polynomials of degree $q-1$ and so can only intersect at most $q-1$ times, there must be some ${\x}\in[q-1]$ such that $g_i\in E_{{\x}}$ but $g\notin E_{{\x}+1}$; w.l.o.g. assume $g_i\in B_{{\x}+1}$. Then since $|B_{{\x}+1}|$ is even, it must be that either there exists some $\alpha\in A_{\x}$ that is in $B_{{\x}+1}$ or there exists some $\beta\in B_{{\x}}$ that is in $E_{{\x}+1}\cup A_{{\x}+1}$. 
If there is such an $\alpha$, then it drops below $f_i$  somewhere between columns ${\x}$ and ${\x}+1$, which means that there is an intersection between columns that can be associated with $(f_i,g_i)$.
Similarly, if there is such a $\beta$, it would intersect $g_i$ somewhere between columns ${\x}$ and ${\x}+1$.
\end{proof}

\begin{theorem}
A \suitable{k} $(q,d)$-curve family is \indt{k} if $$d\ge\left\lceil2\frac{q-1}{2q-1}(k-1)\right\rceil(q-1)+1.$$
\end{theorem}
\begin{proof}
Let $d=\left\lceil2\frac{q-1}{2q-1}(k-1)\right\rceil(q-1)+1$, and suppose for contradiction that there is a bad $(q,d,k')$-arrangement \arr{S} for some $k'\le k$.

Partition the first $d-1$ columns into consecutive groups of $q-1$ columns.
Each column has at least $k'/2$ intersections on it, and each group has a further $\left\lceil\frac{k'}{4}\right\rceil$
intersections either on the columns, between them, or between the last column of the group and the next column (this follows from taking a copy of \arr{S} transposed so that the first column of the group becomes column 0, and applying the last lemma to that arrangement).

Therefore, the total number of intersections in \arr{S} is at least
\begin{align*}
\left\lceil2\frac{q-1}{2q-1}(k'-1)\right\rceil\left((q-1)\frac{k'}{2}+\left\lceil\frac{k'}{4}\right\rceil\right)+\frac{k'}{2}> \frac{(q-1)(k')(k'-1)}{2}=(q-1)\binom{k'}{2}
\end{align*}
which is impossible since there can be at most $(q-1)\binom{k'}{2}$ intersections between members of a set of $k'$ polynomials of degree $\le q-1$.
\end{proof}

\bibliographystyle{plain}
\bibliography{ref}

\appendix
\newpage
\noindent{\LARGE\bfseries Appendix}
\section{Experiments}
We created implementations in C of $(2,d)$-curve hash classes (mapping 32 bit keys to 32 bit hashes) for several small values of $d$, using tables filled with pseudorandom numbers. We compared these against against our implementations of both the 2- and 4-character versions of Thorup and Zhang's scheme.

In implementing each version of Thorup and Zhang's scheme, we choose as the matrix to use in deriving characters a Vandermonde matrix with a first row consisting of 1's and a first column (excluding its intersection with the first row) consisting of 0's (such a matrix was one of those suggested in \cite{4uni}). The first derived character was therefore in every case the first input character.
To perform the finite field multiplications required by their scheme, we used precomputed multiplication tables. The 4-character implementation makes use of the optimization of having parallel additions of 8-bit values stored as segments of larger integers, but the 2-character one does not as it did not seem as though that gave an advantage.

For each hash class, the following test was made thirty times: A member of the class was chosen (i.e. the tables were filled with (pseduo)random values), an array of $1000000$ (pseudo)random keys was generated, and the array was iterated through 10 times, with each entry being hashed each time by the chosen function. The mean time (in nanoseconds) for each application of the function was calculated by averaging the average time for each of the thirty tests. The sample standard deviation between tests was also recorded.

The following table shows these timing results. For integer $d$, curve2\_$d$ is a function from a $(2,d)$-curve family; tz2\_$d$ is an implementation of Thorup and Zhang's $q=2$ scheme with $d$ derived characters, and tz4\_$d$ is an implementation of their $q=4$ scheme with $d$ derived characters. For comparison, we also tested the identity function id.

Three computers were used for testing, ``crunch'', ``laptop'', and ``desktop''.

\begin{figure}[h!tbp]\footnotesize
\begin{tabular}{|r|l|cc|cc|cc|}
\hline
\multirow{2}{*}{guaranteed $k$} & \multirow{2}{*}{function} & \multicolumn{2}{|c|}{time on crunch} & \multicolumn{2}{|c|}{time on laptop} & \multicolumn{2}{|c|}{time on desktop} \\
& & mean (ns) & SD & mean (ns) & SD & mean (ns) & SD\\\hline
0 & 
id
 & 2.08 & 0.0029
 & 2.28 & 0.6609
 & 4.21 & 0.0027
\\
\hline
7 & 
curve2\_4
 & 13.67 & 0.4724
 & 109.93 & 0.6524
 & 72.93 & 8.7217
\\
7 & 
tz2\_6
 & 19.73 & 0.1768
 & 233.40 & 57.0687
 & 104.85 & 4.9674
\\
7 & 
tz4\_16
 & 22.71 & 0.0367
 & 44.15 & 0.7651
 & 56.62 & 0.0888
\\
\hline
9 & 
curve2\_5
 & 29.53 & 0.6552
 & 155.91 & 0.8912
 & 112.49 & 0.7536
\\
9 & 
tz2\_8
 & 37.19 & 2.0565
 & 328.88 & 69.3904
 & 184.06 & 1.9557
\\
9 & 
tz4\_22
 & 33.81 & 0.0464
 & 60.47 & 0.9717
 & 81.80 & 0.4511
\\
\hline
11 & 
curve2\_6
 & 49.37 & 0.3285
 & 195.65 & 1.8137
 & 154.21 & 0.8712
\\
11 & 
tz2\_10
 & 76.91 & 3.9032
 & 432.71 & 64.0526
 & 290.92 & 1.7706
\\
11 & 
tz4\_28
 & 45.38 & 0.0588
 & 73.47 & 0.7742
 & 103.63 & 1.8915
\\
\hline
13 & 
curve2\_7
 & 72.96 & 0.2663
 & 239.33 & 2.8124
 & 202.28 & 1.1247
\\
13 & 
tz2\_12
 & 132.85 & 2.7763
 & 557.84 & 84.8628
 & 398.16 & 1.4015
\\
13 & 
tz4\_34
 & 58.51 & 0.0703
 & 106.74 & 37.9402
 & 120.98 & 0.2935
\\
\hline
15 & 
curve2\_8
 & 96.75 & 0.3063
 & 291.42 & 5.8155
 & 257.54 & 1.0047
\\
15 & 
tz2\_14
 & 187.16 & 2.0494
 & 676.39 & 154.8639
 & 492.78 & 6.6196
\\
15 & 
tz4\_40
 & 69.07 & 0.0772
 & 149.44 & 64.3039
 & 151.02 & 0.1966
\\
\hline
17 & 
curve2\_9
 & 120.73 & 0.2728
 & 339.61 & 4.9303
 & 310.56 & 0.9222
\\
17 & 
tz2\_16
 & 241.09 & 0.9494
 & 870.07 & 251.6200
 & 575.41 & 1.4096
\\
17 & 
tz4\_46
 & 83.55 & 0.1112
 & 192.05 & 61.9977
 & 170.63 & 2.4714
\\
\hline
19 & 
curve2\_10
 & 143.09 & 0.4698
 & 427.99 & 55.5387
 & 365.78 & 1.2003
\\
19 & 
tz2\_18
 & 297.71 & 0.6359
 & 1050.25 & 346.5042
 & 671.39 & 2.0122
\\
19 & 
tz4\_52
 & 97.54 & 0.0715
 & 218.52 & 69.6953
 & 171.57 & 3.2114
\\
\hline
\end{tabular}
\caption{Timing results with sample standard deviations for hash functions that have been proven to be \indt{k} for small values of $k$.}
\end{figure}

Specifications of the computers are as follows: ``crunch'' had two 3 GHz Intel Xeon processors with 4 MB cache each, ``laptop'' had an 1.9 GHz AMD Athlon TK-57 processor with 512 ($2\times 256$) KB cache, and ``desktop'' had a a 3 Ghz Pentium 4 processor with 2 MB cache. Note that the code was compiled for each machine separately, and as a 32-bit executable on ``desktop'' and as 64-bit executables on the others.

We see that in all tests $(2,d)$-curve functions outperformed (often by more than a factor of 2) functions from Thorup and Zhang's $q=2$ class that are known to give as much independence.

However, the tz4\_$d$ functions were often the best (always so on laptop and desktop). On crunch the T
tz4\_$d$ functions did not seem to be sped up as much as the others overall.
A possible reason is that the tz4\_$d$ functions do not benefit so much from a larger cache as the other functions because the derived characters used by tz4\_$d$ are only 8 bits long (so the lookup tables are small).

\end{document}